\documentclass[journal]{IEEEtran}
\usepackage[dvips]{graphicx}
\usepackage{amsmath}
\usepackage{amssymb}
\usepackage{amsfonts}

\usepackage{setspace}

\newtheorem{proposition}{Proposition}

\begin{document}

\title{Joint Power and Resource Allocation for Block-Fading 
Relay-Assisted Broadcast Channels}

\author{
Mohammad~Shaqfeh,~\IEEEmembership{Member,~IEEE,}
Hussein~Alnuweiri,~\IEEEmembership{Member,~IEEE}
\thanks{M.~Shaqfeh and H.~Alnuweiri are with the Department of Electrical 
and Computer Engineering, Texas A\&M University at Qatar,
C/o Qatar Foundation, PO Box 23874, Doha, Qatar.
(E-mail:  \{Mohammad.Shaqfeh, Hussein.Alnuweiri\}@qatar.tamu.edu).}
\thanks{This work was supported by 
the National Priorities Research Program from 
Qatar National Research Fund.}
}

\maketitle

\begin{abstract}
We provide the solution for optimizing the power and resource allocation 
over block-fading relay-assisted broadcast channels 
in order to maximize the long term average achievable rates region
of the users.
The problem formulation assumes regenerative (repetition coding) decode-and-forward (DF) relaying strategy,
long-term average total transmitted power constraint,
orthogonal multiplexing of the users messages within
the channel blocks,
possibility to use a direct transmission (DT) mode from the base station
to the user terminal directly or a relaying (DF) transmission mode,
and partial channel state information.
We show that our optimization problem 
can be transformed into an equivalent ``no-relaying'' 
broadcast channel optimization problem 
with each actual user substituted by two virtual users  
having different channel qualities and multiplexing weights.
The proposed power and resource allocation strategies are expressed in closed-form
that can be applied practically in
centralized relay-assisted wireless networks.
Furthermore, we show by numerical examples that our scheme enlarges the 
achievable rates region significantly.
\end{abstract}


\section{Introduction}
\label{sec:introduction}
Although the relay channel was investigated by the
information theory researchers long time ago (\cite{Va:1971}, \cite{COEL:1979}),
the topic of cooperation/relaying schemes have recently become 
a very active research area within both
the information theory as well as 
communications engineering societies. Few examples of recent works among many others are 
\cite{NiTsWo:2004}, \cite{LaLiEl:2006}, \cite{KrGaGu:2005}, 
\cite{ElMoZa:2006}, \cite{HoZh:2005}, \cite{ReKuVe:2004}, \cite{GuEr:2007}, and
\cite{StEr:2004}. 
It is well-understood now that relaying strategies can 
improve the coverage of wireless networks
by providing higher 
data rates or better transmission reliability to 
terminals at the edge of a wireless cell (i.e.~terminals which 
receive low signal power from the base station).

Relaying technologies are also becoming part of the telecommunication standards \cite{YaHuXuMa:2009}. 
Although many advanced schemes 
based on the cooperation of the mobile users to help each other are being studied in the literature, 
the first actual deployment step which will take place within 
the 3GPP\footnote{Third Generation Partnership Project.}
Long-Term Evolution (LTE)-Advanced (c.f.~\cite{SeToBa:2009}, \cite{DaPaSkBe:2008}, 
\cite{EkFuKaMePaToWa:2006}) standard is based on fixed access points\footnote{The
relays will be base stations but without a wired or microwave connection
to the backhaul network.} 
to do the relaying and within a centralized scheme in which the 
e-nodeB (base station with backhaul connection) takes the scheduling 
and resource allocation decisions. 
One major objective in 3GPP evolution 
is to utilize the scarce wireless system resources efficiently
because achieving the high Quality-of-Service (QoS) targets through over-provisioning 
is uneconomical due to the relatively high cost for 
transmission capacity in cellular access 
networks \cite{Ek:2009}.   

Our objective here is to obtain the optimal (in information theoretic perspective) 
resource allocation schemes but with applying system constraints that 
are relevant to the LTE-Advanced standard so that it can be 
applied practically in the "first introduction" of relays to 
the wireless "cellular systems" industry. We have been able to 
derive the optimal power, resource allocation and scheduling polices 
that are provably based on closed-form formulations 
which are practical for implementation.
Our proposed resource allocation schemes provide an integrated solution
to exploit the cooperative (relaying) diversity gains \cite{NiTsWo:2004} as well as the
multiuser diversity gains \cite{TsVi:05}.

Optimal dynamic resource allocation over fading channels has been investigated 
in the literature for non-cooperative (i.e.~with no relaying) wireless systems in \cite{Gal:68} 
for the single-user case, in \cite{TsHa:98} for multi-access (many-to-one) 
channels, and in \cite{LiGo1:2001} and \cite{Ts:98} for 
broadcast (one-to-many) channels.
Optimal resource allocation for fading relay channels has been 
studied in \cite{LiVe:2005}, \cite{LiVeVi:2007} for single source and single
destination node case.
Resource allocation for broadcast-relay channels was also 
treated in \cite{ReKuVe:2005} and \cite{LiVe:2007}, where the problem setup 
assumes that users terminals cooperate together (i.e.~act as relays to help
other users). This is a different problem setup than the problem
in this paper in which we assume that fixed-relays 
are used to assist in the transmission without being destination
nodes themselves.
In \cite{MeDa:2008} and \cite{MeDa2:2008}, joint power and resource allocation
over multiple access relay-assisted channels was studied for
a constant channel realization scenario, where it was demonstrated 
that joint allocation of power and channel
resources can enlarge the achievable rate region.

In this paper, we consider the broadcast relay-assisted channels
under block-fading conditions (i.e.~over many channel
realizations). In \cite{ViMu:2005}, \cite{KaPoHa:2009}, 
and \cite{SaAdRaYaFaKi:2010} relay-assisted broadcast (downlink)
channels with centralized scheduling were considered with
different performance metrics such as throughput-guarantees
or fairness measures. In our work, we tackle the problem
from an information-theoretic perspective in which we aim to maximize
the achievable rate region. This is equivalent to the problem of
minimizing the transmitted power to achieve requested rate demands.
To the best of our knowledge, maximizing the achievable rate region
of block-fading relay-assisted broadcast channels 
for the case of applying half-duplex regenerative decode-and-forward (RDF) 
\cite{NiTsWo:2004} relaying strategy and orthogonal
multiplexing of user messages within the channel
blocks has not been treated in the literature.
As well-known form the information theory literature,
the achievable rate region can be enlarged (i.e.~improved)
by using non-orthogonal (superposition-based with successive interference
cancellation at the receivers) transmission strategies
as well as more advanced relaying strategies\footnote{For example, 
full-duplex relaying strategies, which require the relay terminal 
to be able to receive and transmit simultaneously on the same channel,
gives better throughput gains than half-duplex strategies. 
However, this requirement is difficult to implement in practice 
because of the complexity of providing electrical isolation 
between the transmitter and receiver circuitries. 
Thus, half duplex relaying strategies are favorable 
from practical point-of-view although they result in 
multiplexing loss because the same message is transmitted 
over two orthogonal slots.}
than RDF.
However, we restrict our
optimization problem with systems constraints
which are favorable for a practical deployment.

We formulate our optimization problem assuming partial 
channel-state-information (CSI). 
Similar to the definition adopted in \cite{GuEr:2007}, 
by partial CSI we mean that the transmitter (source node) knows
the channel state amplitudes only (i.e.~the channel power gain),
while the receiver knows both the channel amplitude and phase.
This can be usually accomplished by separate low rate feedback channels.
According to LTE specifications
\cite{DaPaSkBe:2008}, CSI is obtained by various means and forwarded
to the base station, where the resource allocation decisions are 
done. The scheduling information are provided to the
other nodes in the network using dedicated control channels.
Furthermore, we apply a single long-term
average total (i.e.~sum) power constraint in the formulation of 
the optimization problem 
instead of using separate power constraints
for every involved node (the source and relay nodes).
This approach has been adopted in some works in the literature such as \cite{GuEr:2007}.
The main reasons for applying a single power constraint are:
(i) using separate power constraints for every node is less dynamic 
than using a sum-power constraint for all nodes since the solution 
of the latter formulation involves the optimal power distribution among 
the nodes to maximize the required objective (maximizing the achievable rate region 
in our case). On the other hand, with separate power constraints, the power 
distribution among the nodes is fixed beforehand and, as a result, 
we lose one factor that can help us to maximize our objective, 
(ii) we believe that the sum-power constraint is relevant in practice since the 
relay nodes are fixed access points that are not limited in energy supply 
(i.e. they are not running on batteries like mobile handsets). 
So, we can distribute the power among the relays based on the 
users' distribution in the cell, and hence we can involve one relay more 
than the others if there are more users in its vicinity. 
Note that if no relays are involved at all,
the base station will transmit all the power anyway. In our
problem formulation, we are just
re-distributing some of that total power among the relays
in the most efficient way,
(iii) solving the optimization problem with many power constraints
will be based on
bisection-based methods over many dimensions similar to the 
solution in \cite{MeDa:2008} for the multiple-access case,
while using a sum power constraint will result in  
closed-form solutions which are simpler and more practical
to be applied in real systems.

Following this introduction section, we present 
in Section~\ref{sec:problem} the channel model and the optimization 
problem formulation as well as the applied mathematical notation.
In Section~\ref{sec:solution} we provide the solution of the optimization 
problem. Despite the relative large number of optimization variables involved in 
the problem, we demonstrate that the problem can be solved by: 
(i) characterizing the maximum achievable rates using RDF links with 
optimal power allocation over the source and the relay, 
and (ii) transforming the problem into an equivalent ``no-relaying'' 
broadcast channel with each ``actual'' user replaced by two ``virtual'' users  
having different channel qualities and multiplexing weights.
In order to understand the mathematical solution steps,
the reader is advised to go through \cite[Section~III-B]{LiGo1:2001}
in which the orthogonal ``non-cooperative'' broadcast channel
was considered.
We give the solution for our problem which includes 
the closed-form policies to 
select the best relay, schedule the users across the resource units, choose the optimal transmission 
mode and control the transmission power.
We include also the case of optimizing resource 
allocation over each channel realization separately and use it for 
comparison with the optimal case.
We provide numerical examples in Section~\ref{sec:numerical} 
to demonstrate the advantages of our suggested resource allocation scheme.
Then, we summarize the main conclusions of our work 
in Section~\ref{sec:Conclusion}.

\section{System Model and Problem Formulation}
\label{sec:problem}

\subsection{Channel Model}

\begin{figure}[htb]
 \begin{center}
   \includegraphics[scale=0.33]{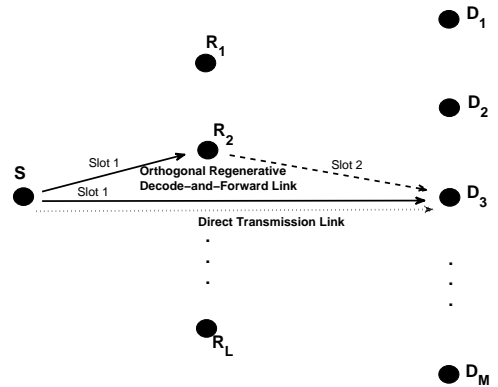}
\caption{System diagram of relay-assisted broadcast channels
with one source node, $L$ relay nodes and $M$ destination
nodes. The message sent to a destination node can go through
half-duplex RDF link with the assistance of one of the relay nodes,
or through a direct transmission link from the source node
to the destination node without the assistance of the relays.}
   \label{fig:system_model}
\end{center}
\end{figure}

We consider the $M$-user relay-assisted broadcast system
shown in Fig.~\ref{fig:system_model}. 
The system consists of one source node ($S$), $M$ destination
nodes ($D$), and $L$ relay nodes ($R$) which can be used
to assist in the transmission from the $S$-node to any of the 
$D$-nodes. The $R$-nodes are connected 
to other nodes by wireless links only. 
Any message transmitted by the 
$S$-node is destined to only one $D$-node.
We assume that the message transfer from the $S$-node
to a $D$-node can go through two possible modes: (i)
a direct transmission (DT) link without the assistance
of any of the $R$-nodes, (ii) a half-duplex regenerative (i.e.~repetition coding) 
decode-and-forward (DF) link in which one of the relays 
assists in the communication between the source
and the destination.
The $S$-node represents the base station, while the $D$-nodes correspond to
the mobile users' terminals.

We assume a block-fading channel model in which the air-link
resource grid is divided in both time and frequency domains
into small blocks called
resource units (RUs)\footnote{As well-known from LTE specifications, 
the time-frequency resource grid is divided into resource 
units (RU) which are allocated flexibly to  
user terminals \cite{EkFuKaMePaToWa:2006}. 
A pair of adjacent resource units (0.5 ms each) is allocated 
to one terminal based on the scheduler decisions.}. The channel is assumed to be constant
within one RU, but varying (fading) independently 
across the RUs in the air-link grid.
The resource units are orthogonal to each other
(non-overlapping). Furthermore, we assume that
all messages are multiplexed orthogonally in all
RUs. Irrespective of their transmission mode (i.e.~DF or DT),
one or many users can receive in the same 
RU (subject to optimization).
However, if more than one user
is receiving messages in the same resource unit, 
the messages of the users are multiplexed
orthogonally by frequency division\footnote{As known 
from LTE specifications, one RU spans several
OFDM subcarriers. So, in principle,
we can let multiple users share the same RU orthogonally.
However, according to LTE specifications, 
one RU can be used by one
user only. So, it may appear that our problem formulation
is not practical to be applied in LTE systems.
However, as discussed in Section~\ref{sec:solution},
the solution of the optimization problem with the orthogonality
constraint involves that only one $D$-node receives in one channel
block (RU unit) using either DT or DF.}. A DF link is divided 
into two sub-units occupying
the same frequency band but having orthogonal time division
multiplexing.  
Fig.~\ref{fig:RU} shows how
the air-link resource grid is divided,
and an example of how the users may be scheduled
across the RUs. 
One of the possibilities is that two independent messages
are sent to the same $D$-node in the same RU, 
but with different transmission modes;
i.e.~one message is sent through a DT link
and the other one using a DF link.
Several other possibilities are illustrated in Fig.~\ref{fig:RU}.

\begin{figure}[htb]
 \begin{center}
   \includegraphics[scale=0.4]{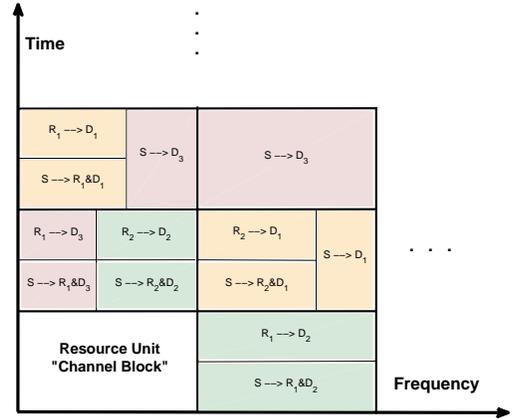}
\caption{Dividing the total air-link resource grid into small blocks
called resource units, which are assumed to be flat faded (i.e.~the 
bandwidth of a RU and its time duration are smaller than the coherence
bandwidth and the coherence time of the fading channels).
Within one RU, one or more users can receive a message from the source
node. All transmissions are orthogonal to each other. Based on the channel 
conditions, the messages and their associated transmission modes can be scheduled 
flexibly to the users terminals.}
   \label{fig:RU}
 \end{center}
\end{figure}

Additionally, we assume that the channel realizations (blocks) are 
large enough (i.e.~slow fading assumption) so that each codeword
destined to any user 
can be transmitted over one channel block with close-to-capacity
limits rate.

\subsection{Problem Formulation}
\label{sec:problem_definition}
In this work, we provide answers to several fundamental 
questions\footnote{Some of these questions may have been discussed in
the literature for other channel models or problem setup. We are interested here in
the optimal resource allocation scheme within the assumptions and
channel model adopted in this paper.} concerning the resource allocation schemes 
to be applied in the system:
\begin{enumerate}
\item When (i.e.~under what conditions) 
is a DT link optimal? When is a DF
link optimal? When is the orthogonal multiplexing 
of both modes optimal?
\item How should we select the relay that is most capable of
supporting the transmission to a $D$-node? What is the optimal policy 
to allocate the power for the $S$-node and the $R$-node over a DF link?
\item How should the scheduler process the CSI to
determine the number of messages to be transmitted in a RU and
their transmission modes? What is the optimal
policy to allocate the power and the channel ratios for every
message?
\end{enumerate}

Furthermore, in a flexible system where the users can have  
different requested rates dependent on, for example, the
supported services, the resource allocation schemes
should be flexibly adjustable to operate at any of the possible
operating points of the system. In this paper, we solve the 
resource allocation problem to operate at any pre-determined
point on the achievable rates region. 
The achievable rates region is defined as the set of 
all long-term average rate 
vectors\footnote{We use boldface to indicate vectors.}
$\mathbf{\bar R}=[\bar R_1 \; \bar R_2 \cdots \bar R_M]$ achieved by the nodes
$D_1, D_2, \cdots, D_M$ such that the long-term average 
sum (of all nodes) power density constraint
$\bar{P}$ is not exceeded.  
To simplify the mathematical formulation of the problem,
we assume that $R$ and $P$ represent spectral densities
(i.e.~bits/sec/Hz and Joul/sec/Hz respectively).

The optimum points within the achievable rates region 
are those that are located on the
boundary surface.  The latter can be characterized as the
closure of the parametrically defined surface
\begin{equation}
\label{eq:surface}
\big\{ \mathbf{\bar R}(\pmb{\mu}):\pmb{\mu}\in\mathcal{R}^{M}_{+},
\sum_i\mu_i=1\big\}
\end{equation}
where for every weighting-factors 
vector $\pmb{\mu}$,
the rate vector $\mathbf{\bar R}(\pmb{\mu})$ can be obtained
by solving the optimization problem: 
\begin{equation}
\label{eq:problem}
\max \sum^{M}_{i=1}\mu_i \bar R_i, \quad 
\mbox{subject to }
\frac{1}{K}\sum^{K}_{k=1}P[k]
\leq \bar{P}
\end{equation}
where $K\rightarrow\infty$ is the total number of channel blocks, and
$M$ is the number of active users.
$\bar R_i$ is the long-term average (i.e.~averaged over all channel blocks) rate of user $i$.
The relationship between $R$ and $P$ will be obtained in Section~\ref{sec:solution}.
For simplicity, we use a single index $k$ to refer to a RU
although the resource grid is divided in both time and frequency domains.
We assume without-loss-of-generality that all channel blocks
(RUs) have identical frequency bandwidth and time duration.
$P_i[k]$ and $P[k]$ are respectively the sum (of the source node
and the relay node) power density (Joul/sec/Hz) used specifically to transmit to node
$D_i$, and the sum power density (Joul/sec/Hz) transmitted (including all receiving nodes) 
during channel block $k$.

All power and resource allocation polices proposed in this paper are presented as functions
of the weighting factors vector because it defines
the specific operating point of the system.
The selection of $\pmb{\mu}$ to meet the constraints of the provided
services (applications) is a different topic that is not discussed 
in this work\footnote{As discussed in \cite{ShGo3:2008} and \cite{ShGoMc:2008},
there is no contradiction between the two objectives of (i) efficient resource
allocation by designing resource allocation schemes
leading to operating at the points on the boundary of the achievable rates
region, and (ii) achieving fairness among the users
as well as maintaining the QoS requirements, which can be done
by controlling the operating point of the system based on proper selection 
of $\pmb{\mu}$.}.
Few examples of the many possible
approaches suggested in the literature to select the specific operating 
point of the system are (i) the fairness-based approach, such as 
the proportional fairness scheduler \cite{ViTsLa:02} and 
the flexible resource-sharing constraints scheduler \cite{ShGo4:2008}, 
(ii) the utility-maximization-based approach \cite{SoLi3:2005}, and
(iii) the Quality-of-Service (QoS) constraints based 
approach \cite{WaGiMa:2007}\cite{WaGi2:2007}.

\subsection{Mathematical Notation}
The notation $R_i[k]$ means the achievable rate (bits/sec/Hz) at
node $D_i$ during one RU, which has index $k$.
In our mathematical notations, we use superscripts to differentiate
between DT and DF links. Similarly, we use subscripts to indicate
the nodes involved.
We denote the relay node associated with destination node
$D_i$ as $R_{l(i)}$.
Furthermore, we denote the channel access ratios (i.e.~bandwidth ratio) of user $D_i$
within channel block $k$ as $\tau_i^{DT}[k]$ for the DT link,
and $\tau_i^{DF}[k]$ for the DF link.
It should be clear that $\sum_{i=1}^M(\tau_i^{DT}[k]+\tau_i^{DF}[k])=1$.  
The channel gain between two nodes is represented by $h$. For
example, $h_{sd_i}$ is the channel quality between the source node and 
the destination node ($D_i$).
Fig.~\ref{fig:system_model2}
shows an example of one RU (with index $k$) with a DF link to one
node $D_i$ as well as a DT link to another node $D_j$ multiplexed orthogonally in 
the frequency domain.
The relations between the transmitted signals ($x_i^{DF}[k]$, $x_j^{DT}[k]$) and the
signals ($y[k]$) received by the destination and relay node associated
with node $D_i$ in each of the orthogonal sub-channels within
the resource unit $k$ are shown in Fig.~\ref{fig:system_model2}. 
Here, $z_{d_j}^{DT}[k]$ and $z_{r_{l(i)}}^{DF_1}[k]$ represent
the additive zero-mean white circular complex Gaussian noise
with variance $\sigma^2$ at $D_j$ and $R_{l(i)}$ respectively.

\begin{figure}[htb]
 \begin{center}
   \includegraphics[scale=0.39]{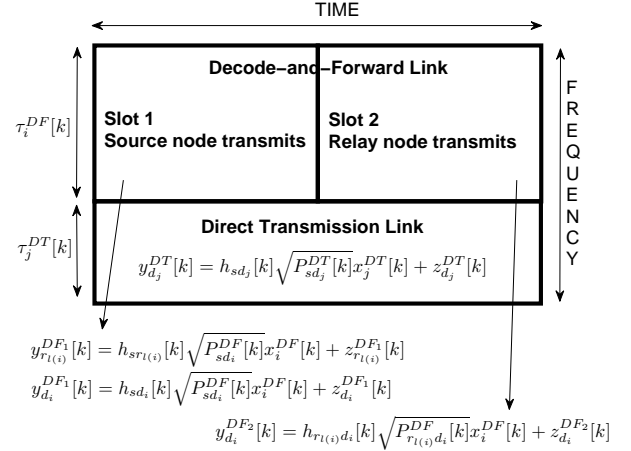}
\caption{An example of the orthogonal division of one 
resource unit $k$ to transmit to users $D_i$ (using
DF mode with assistance of node $R_{l(i)}$) and $D_j$ (using
DT mode). The received signals at the associated
nodes are shown as functions of the transmitted
signals.}
   \label{fig:system_model2}
 \end{center}
\end{figure}

We use the notation $\gamma = \frac{|h|^2}{N_o}$
for the effective power gain of a given channel, where
$N_o$ is the noise power spectral density.
We assume that the fading processes of the channel gains ($\gamma_{sd_i}$, 
$\gamma_{sr_{l(i)}}$, $\gamma_{r_{l(i)}d_i}$) are independent of each other, 
stationary and have continuous probability density functions, $f_\gamma(x)$.
In the numerical examples throughout the paper, 
we assume the fading processes have
Rayleigh\footnote{$f_{\gamma}^{\text{Rayleigh}}(x)=
\frac{1}{\bar{\gamma}}\exp
\left(\frac{-x}{\bar{\gamma}}\right)$, $x>0$, 
$\bar{\gamma}$ is the average effective power gain of the channel.} 
or Rice \footnote{$f_{\gamma}^{\text{Rice}}(x)=\frac{\kappa+1}{\bar{\gamma}}
\exp\left(-\kappa-\frac{\kappa+1}{\bar{\gamma}}x\right)
I_o\left(2\sqrt{\frac{\kappa(\kappa+1)x}{\bar{\gamma}}}\right)$, $x>0$, 
$\bar{\gamma}$ is the average effective power gain of the channel, 
and $\kappa$ is the ratio 
of the power received through the line-of-sight path 
to the power received through the non-light-of-sight path.} distributions \cite{Pro:95}.

\section{Solution Structure}
\label{sec:solution}

Looking back at our optimization problem (\ref{eq:problem}), the
achievable rate $R_i[k]$ is the sum of the rates achieved by the
DT and DF links multiplied by their relative channel access ratios:
\begin{equation}
\tau_i[k]R_i[k]=\tau_i^{DF}[k]R_i^{DF}[k]+\tau_i^{DT}[k]R_i^{DT}[k]
\end{equation}

The maximum possible achievable rate (bits/sec/Hz)
of user $i$ through DT link is given by:
\begin{equation}
\label{eq:rateDT}
R_i^{DT}[k]=\log\left(1+\gamma_{sd_i}[k]P_{sd_i}^{DT}[k]\right)
\end{equation}
for additive white Gaussian receiver noise, where (\ref{eq:rateDT}) is the
Shannon capacity for the AWGN channel. With adaptive modulation and coding,
a rate close to capacity can be achieved (e.g. \cite{Go:2005}).
In practice,
wireless systems support a set of discrete rate values rather than
a continuous range.  However, we use the idealization (\ref{eq:rateDT}) to
relate ``power'' and ``rate''; the relative
performances\footnote{The ``absolute'' performance of a combination of
specific modulation and coding schemes can often be approximated by
(\ref{eq:rateDT}) as well. An ``acceptable'' residual bit or frame
error rate will often be achieved by a practical scheme with some
(fairly constant) power-offset against the theoretical ``zero-error''
curve given by (\ref{eq:rateDT}).} 
will carry over into practice.

The achievable rate by regenerative decode-and-forward relaying
is known from the literature (e.g.~\cite{NiTsWo:2004}):
\begin{align}
& R_i^{DF}[k]= \nonumber \\
& \min\left\{
   \begin{array}{l}
        \frac{1}{2}\log\left(1+\gamma_{sr_{l(i)}}[k]P_{sd_i}^{DF}[k]\right), \\
        \frac{1}{2}\log\left(1+\gamma_{sd_i}[k]P_{sd_i}^{DF}[k]+
\gamma_{r_{l(i)}d_i}[k]P_{r_{l(i)}d_i}^{DF}[k]\right)
    \end{array} 
 \right. \label{eq:onerelay}
\end{align}

The first term in (\ref{eq:onerelay}) is the achievable rate by the
relay node in the first time slot. Since, in RDF, the relay has to decode the
signal in order to retransmit it in the second time slot, the achievable
rate by the destination node is upper bounded by the ``decodable'' rate limit
at the relay. The $1/2$ factor which is multiplying the log function
is due to the fact that the information is conveyed to the 
relay in half the total time assigned for transmitting the signal to 
the $D$-node. 
The second term in (\ref{eq:onerelay})
is the maximum achievable rate by the $D$-node using Maximal-Ratio-Combining (MRC),
e.g.~\cite{TsVi:05}, of the signal received from the source 
in the first slot, and the signal
received from the relay in the second time slot.

In the sequel, we assume a single relay selection per-user and per-resource
unit in RDF transmission mode. This is appropriate from a practical point
of view especially that we assume partial CSI (no phase information) 
available at the transmitters.
In Appendix~\ref{sec:appendix} we extend the solution to the case of
multiple-relay selection per-user which requires full CSI (i.e.~including
phase information) in order to enable coherent (in phase) transmission
of multiple relays.

The set of optimization variables for the problem (\ref{eq:problem})
includes the power allocated to each node as well as the corresponding
channel access ratios for every transmission mode in every channel block $k$.
We need first to characterize the achievable rate
by RDF links in order to solve (\ref{eq:problem}).

\subsection{Achievable Rate by Regenerative Decode and Forward}
\label{sec:RDF}
Our objective here is to maximize $R_i^{DF}[k]$  
(\ref{eq:onerelay}) with respect to the
sum power allocated to the DF link:
\begin{equation}
\label{eq:totalpowerDF} 
P_i^{DF}[k]=\frac{1}{2}P_{sd_i}^{DF}[k]+\frac{1}{2}P_{r_{l(i)}d_i}^{DF}[k]
\end{equation}
The $1/2$ factors in (\ref{eq:totalpowerDF}) are used because we present
$P$ as (Joul/sec/Hz), and the source and relay nodes are transmitting
in half the total time. 
For simplicity, we will drop the user index $i$ as well as the channel
block index $k$ in our mathematical formulations in this sub-section.

We first define when the DF link can be useful. By the 
notion ``useful DF link'' we mean that it could (for some power range)  
support higher data rate than a DT link given that both transmission modes 
are allocated the same total power. 

\begin{proposition} 
If $\gamma_{sr}$ or $\gamma_{rd}$ is less than
$\gamma_{sd}$, then transmitting using a DT link supports higher data rates
than using a DF 
link, which means that the DF link is NOT useful in this case. 
\end{proposition}
\begin{proof}
If $\gamma_{sr}<\gamma_{sd}$, then the relay node can decode at rates
less than the rates achievable by the DT link. Similarly, if $\gamma_{rd}<\gamma_{sd}$,
then allocating the power in the second slot to the source node to retransmit the 
codeword is better than allocating the power to the relay node.
\end{proof}

The next step is to solve the problem: 
\begin{subequations}\label{eq:RDF_opt}
\begin{align}
& \max_{P_s^{DF},P_r^{DF}} 
R^{DF}=\min\left\{
   \begin{array}{l}
        \frac{1}{2}\log\left(1+\gamma_{sr}P_s^{DF}\right), \\
        \frac{1}{2}\log\left(1+\gamma_{sd}P_s^{DF}+\gamma_{rd}P_r^{DF}\right)
    \end{array} 
 \right. \label{eq:RDF}\\
\nonumber&\mbox{subject to }\\
& \gamma_{sr}>\gamma_{sd}, \; \gamma_{rd}>\gamma_{sd}, \; \left(\frac{1}{2}P_s^{DF}+\frac{1}{2}P_r^{DF}\right)= P^{DF}.
\end{align}
\end{subequations}
 
\begin{proposition} 
The optimal allocation of the source power and the relay power over
a useful DF link can be obtained by making the two terms in (\ref{eq:RDF}) equal
each other. 
\end{proposition}
\begin{proof}
The first term in (\ref{eq:RDF}) (i.e.~$\frac{1}{2}\log\left(1+\gamma_{sr}P_s^{DF}\right)$)
is a monotonically increasing function of $P_s^{DF}$.
On the other hand, the second term in (\ref{eq:RDF})
(i.e.~$\frac{1}{2}\log\left(1+\gamma_{sd}P_s^{DF}+\gamma_{rd}P_r^{DF}\right)$)
is a monotonically decreasing function of $P_s^{DF}$ because
$\gamma_{rd}>\gamma_{sd}$, and the sum of $P_s^{DF}$ and
$P_r^{DF}$ is constant (equals $2P^{DF}$). Thus, to maximize the minimum of the
two terms in (\ref{eq:RDF}), we should make them equal.
\end{proof}

Hence, the power allocation over a useful DF link should be as follows: 
\begin{equation}
\label{eq:PsPrDF}
P_s^{DF}=\frac{2P^{DF}}{1+\frac{\acute{\gamma_{sr}}-1}{\acute{\gamma_{rd}}}},
\;
P_r^{DF}=2P^{DF}-P_s^{DF}
\end{equation}
where $\acute{\gamma_{sr}}$ and $\acute{\gamma_{rd}}$ in (\ref{eq:PsPrDF}) 
are defined as: 
\begin{equation}
\label{eq:define_acute}
\acute{\gamma_{sr}} \doteq \frac{\gamma_{sr}}{\gamma_{sd}}, \;
\acute{\gamma_{rd}} \doteq \frac{\gamma_{rd}}{\gamma_{sd}}
\end{equation}

Furthermore, the achievable rate over a useful DF link
(i.e.~$\acute{\gamma_{sr}}>1$
and $\acute{\gamma_{rd}}>1$)
is: 
\begin{equation}
\label{eq:rateRDFlink}
R^{DF}(P^{DF})=\frac{1}{2}\log\left(1+2\alpha\gamma_{sd}P^{DF}\right)
\end{equation}
where
\begin{equation}
\label{eq:alpha}
\alpha=\frac{\acute{\gamma_{sr}}\acute{\gamma_{rd}}}
{\acute{\gamma_{sr}}+\acute{\gamma_{rd}}-1}
\end{equation}

The parameter $\alpha$ is actually the power gain that
the RDF link is capable of providing. It should be clear
from (\ref{eq:alpha}) that if a relay is useful, then $\alpha>1$.
The result in (\ref{eq:rateRDFlink}) is interesting because it
fits with the expectation that a half-duplex relaying strategy 
provides a power gain (i.e.~beamforming gain) at the cost
of a loss (by half) in the degrees-of-freedom\footnote{Readers who are 
not familiar with the notions of degrees-of-freedom, bandwidth-limited
and power-limited  regions of the AWGN channel capacity formula 
can refer to \cite[Chapter~5]{TsVi:05}.} 
(i.e.~multiplexing gain). 
The beamforming gain is obtained because the receiver decodes
the signal at higher Signal-to-noise-ratio (SNR) although
the same total power is used for the transmission.
The loss in the multiplexing gain is a result of transmitting the message
over two time slots.
Thus, RDF, and all 
half-duplex relaying strategies in general, can be useful 
for terminals which are operating in the power limited region 
of the channel capacity where the achievable rate (bits/sec) 
has almost a linear 
relation with SNR. This means that RDF can support higher rates than a DT link 
only for users with low SNR (i.e. close or at the cell edge). 
On the other hand, terminals operating in the bandwidth-limited region
of the channel capacity (i.e.~at high SNR) will be 
affected by the loss in degrees-of-freedom of RDF, 
and hence a DT link can support higher data
rate in this case.
Actually, it can be shown that other half-duplex
relaying strategies such as amplify-and-forward (AF) 
have achievable rate performance
similar to (\ref{eq:rateRDFlink}), but with more complicated
formulas to express the power gain $\alpha$ in terms of the
channel qualities between the source, relay and destination nodes.
RDF has the nice property that $\alpha$ is a function of the 
channel qualities only and independent of the total power allocated 
to the link. This 
is not the case with other relaying strategies such as AF.  

Now, we can define the criteria to select the unique relay that
should be selected in the RDF link. It is the one that
provides the best power gain in order to maximize
the total achieved rate at the destination node given
the total power allocated to the DF link.
The best relay to be associated with node $D_i$ is $R_{l(i)}$
where $l(i)$ is defined as:
\begin{equation}
\label{eq:selectrelay}
l(i) = \arg \max_j\alpha_{r_jd_i} \quad :\acute{\gamma_{sr_j}}>1
\text{ and } \acute{\gamma_{r_jd_i}}>1
\end{equation}

In the remaining part of this paper, we will use 
(\ref{eq:rateRDFlink}) to characterize the performance
of RDF links. We assume that proper relay selection 
is used (\ref{eq:selectrelay}), and based on
the allocated power to the RDF link ($P^{DF}$),
the source and relay power should be allocated according
to (\ref{eq:PsPrDF}).
For simplicity, we will drop the index of the relay
and just use $\alpha_i[k]$ to indicate the best power gain
for the RDF link. 
If none of the relays is useful, we will use the DT link
only for transmission. Otherwise, an orthogonal multiplexing
between the DT link and the DF link should be 
used (subject to optimization). 

\subsection{Optimal Resource Allocation -- Broadcast Channel with $2M$ Virtual Users} 
\label{sec:optimal}
The results in Section~\ref{sec:RDF} leads to a very interesting
consequence that we can transform our problem into 
a broadcast (with no relays) channel with 
each actual user replaced by two virtual users  
having different channel qualities and multiplexing weights, 
corresponding to the two different
transmission modes (i.e.~DT or DF) of the actual user. 
The only difference between our problem and the original
broadcast channel is that we will have two classes of users
in terms of the relation between the achievable rate and the allocated power.
However, this does not change the structure of the problem solution
and the interesting closed-form solution of it \cite{LiGo1:2001}.
The achievable rates multiplied by the corresponding multiplexing factor 
of the $2M$ virtual users\footnote{If some users do not have useful
relay links, then the total number of virtual users will be less than $2M$.} 
will have
the form
\begin{equation}
\label{eq:fj}
f_j(P_j)=\omega_j\log(1+\eta_jP_j) 
\end{equation}
where the virtual user $j$ that is related to one of the
transmission modes of the actual user $i$ is characterized by:
\begin{equation}
\label{eq:omegaeta}
\omega_j=\left\{\begin{array}{ll}
	\mu_i   & \text{if DT}\\
	\frac{\mu_i}{2}  & \text{if DF}
	\end{array}
\right.
,
\quad
\eta_j=\left\{\begin{array}{ll}
	\gamma_{sd_i}   & \text{if DT}\\
	2\gamma_{sd_i}\alpha_i  & \text{if DF}
	\end{array}
\right.
\end{equation}
where $\mu_i$ is defined in the problem formulation \eqref{eq:problem}, 
and $\alpha$ can be obtained using \eqref{eq:alpha} and \eqref{eq:selectrelay}.
As well known \cite{LiGo1:2001},
if the cumulative-density-function (CDF) of the fading process is a continuous
function (this is true in reality such as in Rayleigh and Rice
fading conditions), then \emph{the optimal resource allocation policy is unique
and at most a single user only is scheduled
in each fading state (i.e.~channel block).
Furthermore, the power allocated to the scheduled user follows a water-filling
approach.}

Thus, every RU should be allocated to only one user with only one transmission 
mode (i.e.~either DF or DT and not both).
The index $m$ of the only virtual user that should be scheduled in channel block
$k$ (i.e.~$\tau_l[k]=1$ if $l=m$ and $\tau_l[k]=0$ if $l\neq m$, $l=1,2,\cdots,2M$) is:   
\begin{equation}
\label{eq:functionm}
m = \arg \max_j \left(f_j(P_j[k])-\lambda_G P_j[k] \right)
\end{equation}
where $f_j$ is defined in \eqref{eq:fj}, and 
$\lambda_G$ is the ``power price''\footnote{The optimal 
value of $\lambda_G$ in order to achieve a desired average power level
is dependent on the channel statistics. One option is to control
$\lambda_G$ in real-time based on actual channel measurements \cite{WaGiMa:2007}.} which should be controlled to 
maintain the average power constraint $\bar{P}$. $P_j[k]$ in (\ref{eq:functionm})
is dependent on the transmission mode of the corresponding actual user 
$i$ related to the virtual user $j$:
\begin{equation}
\label{eq:powerm}
P_j[k]=	\left[\frac{\omega_j}{\lambda_G}-\frac{1}{\eta_j[k]}\right]^+ 
\end{equation}
where ($x^+=\max(x,0)$).
Once we obtain $m$ according to (\ref{eq:functionm}), we can determine
the actual user that should be scheduled and the optimal transmission mode.

Expressions \eqref{eq:selectrelay}, \eqref{eq:functionm}, 
\eqref{eq:powerm} and \eqref{eq:PsPrDF} provide valuable closed-form policies
to select the best relay, schedule the users across the RUs, choose the optimal transmission 
mode and control the transmission power.


\subsection{Optimal Resource Allocation with Constant Power 
per Channel Block} 
\label{sec:suboptimal}
If the resource allocation is optimized for each channel block 
independently from other channel blocks and irrespective 
of the scheduled users or selected transmission mode in each 
channel block, the corresponding optimization problem is:
\begin{subequations}\label{eq:problem2}
\begin{align}
& \max_{\{\tau_j[k],P_j[k],j=1,\cdots,2M\}}\sum^{2M}_{j=1}\tau_j[k]f_j(P_j[k]),\\
\nonumber&\mbox{subject to }\\
& \sum^{2M}_{j=1}\tau_j[k]P_j[k] \leq P[k].
\end{align}
\end{subequations}
where $M$ is the number of active users, and 
$f_j$ is the achievable rate multiplied by the weighting index as defined in \eqref{eq:fj}.
Note that since we assumed in our problem definition in Section~\ref{sec:problem_definition} that
$R$ and $P$ represent spectral densities (i.e. their units are bits/sec/Hz and
Joul/sec/Hz respectively), the actual short-term average rate and power 
(averaged within one channel block $k$) should be obtained by multiplying with 
the corresponding channel access ratio (within channel block $k$) $\tau_j[k]$ 
for each virtual user $j$.

The solution of \eqref{eq:problem2} can be obtained by applying the concept of virtual users over 
a no-relaying broadcast channel that is 
discussed in Section~\ref{sec:optimal}.
The solution of the problem \cite{LiGo1:2001} is that either one 
or two virtual users (could be related to the same actual user or 
to two different actual users) are scheduled within the channel block $k$.

Depending on the value of $P[k]$, the solution of the general case
of problem (\ref{eq:problem2}) is either one or two elements in $\mathbf{P}$
and $\mathbf{\tau}$ are non-zero.
The domain $[0,\infty)$ of the functions $f_j$, $j=1,\cdots,2M$ 
will be divided into adjacent intervals which are
alternating between (i) intervals in which the solution of the optimization
problem is that the optimization variables of only one function (i.e.~virtual user)
should be non-zero, given that $P[k]$ belongs to the corresponding interval,
and (ii) intervals in which the optimization variables of two functions
are non-zero, with the values of the $P$ variables of these two functions identical to 
the values of the two end-points of the corresponding interval, and 
the $\tau$ variables for
these two functions dependent on the relative location of $P[k]$
with respect to the interval end-points.  
An algorithm was provided in 
\cite[pp. 1089, first column]{LiGo1:2001} to obtain the solution.

Fig.~\ref{fig:optimizationproblem2} provides an illustration of the 
solution to problem (\ref{eq:problem2}). The example presented in 
Fig.~\ref{fig:optimizationproblem2} is for three virtual users (two actual users 
with one of them having no useful DF link).

\begin{figure}[htb]
 \begin{center}
    \includegraphics[scale=0.38]{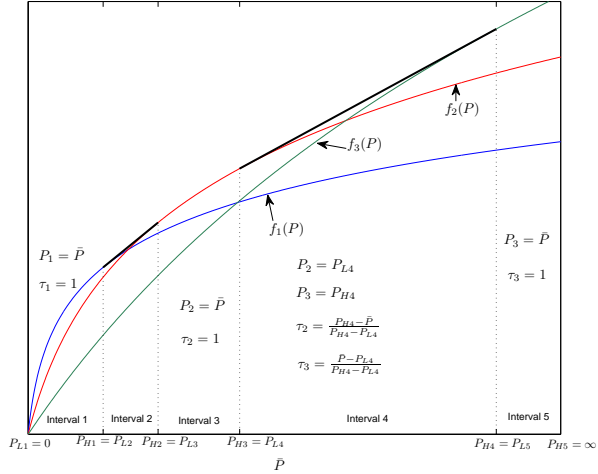}
\caption{An illustrative example of problem (\ref{eq:problem2}) 
with three virtual users.
The dependence of the solution on the specific interval, to which
$\bar{P}$ belongs, is illustrated.}
   \label{fig:optimizationproblem2}
 \end{center}
\end{figure}

In addition to the optimal solution, we suggest also a sub-optimal
solution which is simple and
close to optimality. The near-optimal solution is: 
\begin{equation}
\label{eq:suboptimalsolution2}
\tau_m=1, \; P_m=P[k], \; \text{where } \; m=\arg\max_i f_i(P[k])
\end{equation}

This near-optimal solution is identical to the optimal solution when
only one virtual user is optimizing the problem. However, if two virtual users 
are involved in the optimal solution, the sub-optimal solution
still gives near-optimal performance because the difference
between the tangent line between the two functions, which optimize the problem,
and the maximum of the two functions is usually very small.

\section{Numerical Examples}
\label{sec:numerical}

\begin{figure}[htb]
 \begin{center}
   \includegraphics[scale=0.47]{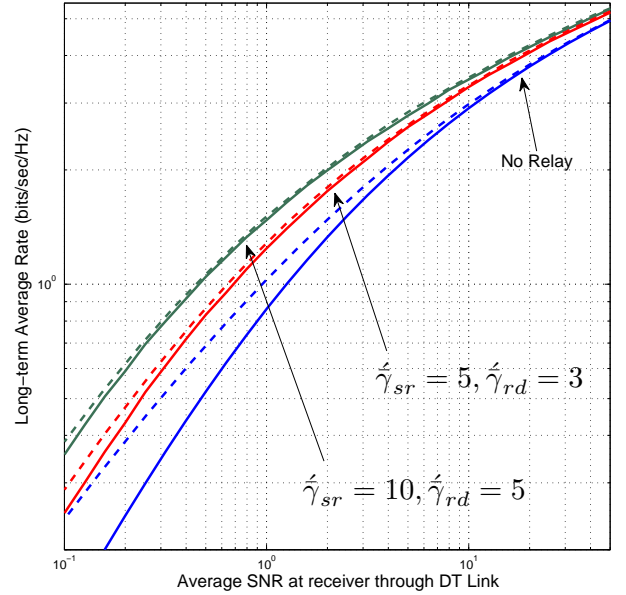}
\caption{Comparisons of the achievable rates using 
(i) direct transmission only,
(ii) orthogonal multiplexing of RDF and DT with
optimal power allocation.
The solid lines correspond to the case of constant total power 
per channel realization, and the dashed lines correspond
to the case of optimal power allocation over all channel realizations.}
   \label{fig:sim_single}
 \end{center}
\end{figure}

We provide in Fig.~\ref{fig:sim_single} numerical comparisons of the
maximum achievable rates in a single-user case with and 
without applying the optimal
power allocation strategy studied in this paper.
The simulation results
were obtained for the scenario when the power gain over
the source-destination link is Rayleigh faded,
while the power gain over the source-relay and
relay-destination is Rician faded with better average channel
qualities than over the direct link. 
Such scenario is relevant when there is no line-of-sight (LOS)
path between the source and the destination nodes, while
the relay node is in a position where it has LOS
paths to both of them.
The simulation results in Fig.~\ref{fig:sim_single}
were done for two cases: 
(i) $\frac{\bar\gamma_{sr}}{\bar\gamma_{sd}}=5$,
$\frac{\bar{\gamma_{rd}}}{\bar{\gamma_{sd}}}=3$, and
(ii) $\frac{\bar\gamma_{sr}}{\bar\gamma_{sd}}=10$,
$\frac{\bar\gamma_{rd}}{\bar\gamma_{sd}}=5$.  
In the simulations, the used values for the ratios 
between the power of the LOS path
over the non LOS paths in the Rician faded channels
are: $\kappa_{sr}=10$ and $\kappa_{rd}=5$.

Fig.~\ref{fig:sim_single} shows that potential gains can be
obtained with the assistance of the relay.
The gain in the achievable rate is higher at lower
SNR, but still useful at mid to high SNR. 
Furthermore, we observe that using optimal power allocation
over all channel realizations has a small gain over
constant power allocation per channel block, especially
at higher SNR.

\begin{figure}[htb]
 \begin{center}
   \includegraphics[scale=0.47]{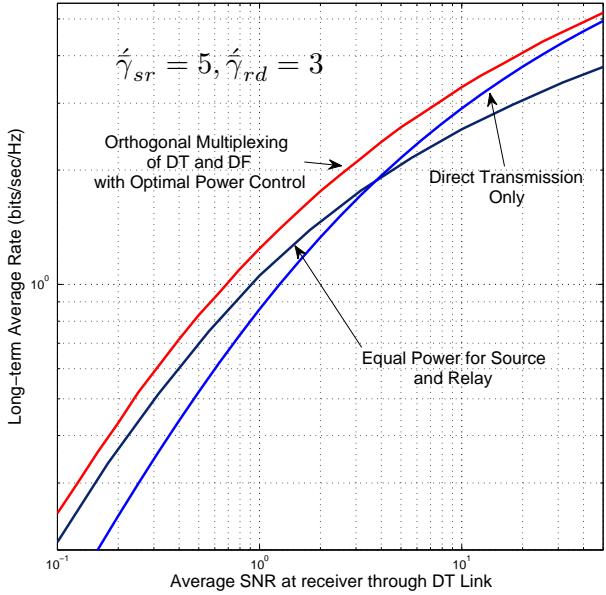}
\caption{Comparisons of the maximum achievable rates with relay assistance 
for two cases: 
(i) optimal power allocation,
(ii) equal power allocation for the source and the relay node.}
   \label{fig:sim_single_b}
 \end{center}
\end{figure}

Another important observation is that with equal power
allocation between the source and the relay nodes,
which is shown in Fig.~\ref{fig:sim_single_b},
there is a significant degradation
from the achievable rates with the assistance of the relay.
Furthermore, it is worse than just using direct transmission
when the SNR is mid to high range.
This demonstrates the advantage of the proposed power
allocation strategy because it switches between
DF and DT based on the channel conditions, and thus
it is always better than just DT over all SNR regions.

\begin{figure}[htb]
 \begin{center}
   \includegraphics[scale=0.39]{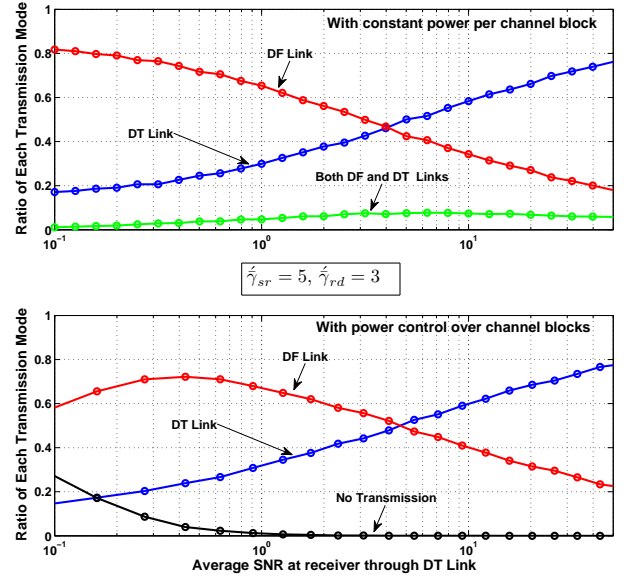}
\caption{Ratio of each transmission mode for the case
of ($\frac{\bar\gamma_{sr}}{\bar\gamma_{sd}}=5$,
$\frac{\bar{\gamma_{rd}}}{\bar{\gamma_{sd}}}=3$) in
Fig.~\ref{fig:sim_single}.}
   \label{fig:mode_percentage}
 \end{center}
\end{figure}

The results in Fig.~\ref{fig:sim_single} and Fig.~\ref{fig:sim_single_b}
were obtained
assuming large number of channel blocks (realizations) in order to obtain
good estimate of the average performance of the system.
Fig.~\ref{fig:mode_percentage} shows the ratio
of RUs -- with respect to the total number
of RUs used in the simulation -- 
which have a relay link, a direct transmission link,
or none for the  
case of ($\frac{\bar\gamma_{sr}}{\bar\gamma_{sd}}=5$,
$\frac{\bar{\gamma_{rd}}}{\bar{\gamma_{sd}}}=3$) in
Fig.~\ref{fig:sim_single}.
The results in Fig.~\ref{fig:mode_percentage} include
the global power control case as well as the constant
power per channel block case. 
It is demonstrated that the DF link is more
used at low SNR and decreases gradually as SNR improves,
while at high SNR, the DT link
is more capable of providing higher throughput
to the D-node. The ``no transmission'' case
appears when the channel condition is so bad that
the optimal power control strategy is not to transmit during
such deep fading conditions. It is obvious that such case
appears more frequently at low SNR.

\begin{figure}[htb]
 \begin{center}
   \includegraphics[scale=0.4]{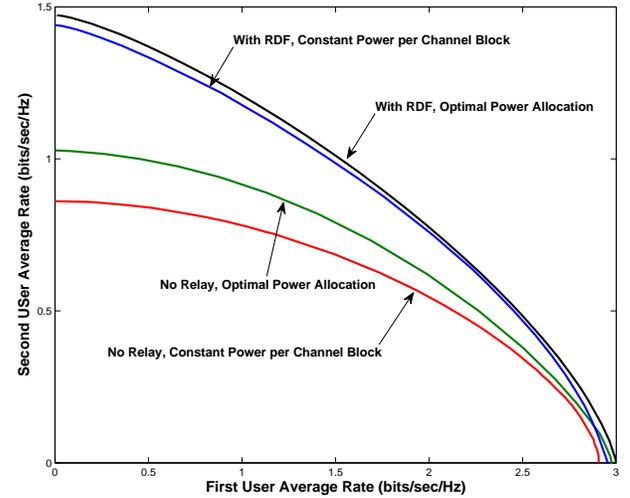}
\caption{Comparisons of the achievable rates in a two-user
case with and without the assistance of a relay, and
with and without power control over
the channel blocks. The links parameters are:
$\bar\gamma_{sr}=10$ (Rician faded with $\kappa_{sr}=10$),
$\bar\gamma_{sd_1}=10$ (Rayleigh faded),
$\bar\gamma_{rd_1}=2$ (Rician faded with $\kappa_{rd_1}=2$),
$\bar\gamma_{sd_2}=1$ (Rayleigh faded),
$\bar\gamma_{rd_2}=5$ (Rician faded with $\kappa_{rd_2}=5$).}
   \label{fig:sim_2user}
 \end{center}
\end{figure}

Fig.~\ref{fig:sim_2user} displays numerical comparisons of the
achievable rates region in a two-user case 
with and without the assistance of a relay
for the two cases of optimal
power allocation over all channel blocks,
and constant power per channel block.
The simulation results
were obtained using the channel quality parameters
shown in the Figure caption. 
The results in Fig.~\ref{fig:sim_2user} demonstrates that 
the RDF strategy improves the performance of the system,
especially for users who have
bad connection with the base station.
Furthermore, we observe that using optimal power allocation
over all channel realizations is advantageous,
especially at low SNR. However, using
constant power allocation per channel block
has good performance at higher SNR.

\section{Conclusion}
\label{sec:Conclusion}
We have addressed the problem of optimal resource allocation 
over block-fading relay-assisted broadcast channels 
under practical system constraints including using orthogonal 
multiplexing of transmitted messages, and regenerative 
decode-and-forward relaying.
We have formulated the optimization problem 
using a sum power constraint, which is a valid assumption 
in fixed-relays scenarios with no power limitations, 
and we have shown that 
this is sensible in order to achieve the maximum possible 
performance as well as to transform the problem into an equivalent 
``no-relaying'' broadcast channel optimization which has simple 
closed-form solution that is practical to be applied in real-time systems.
The optimal resource allocation strategy is to schedule 
only one user within one channel block (i.e.~resource unit) 
and to transmit using either a direct transmission mode 
without the assistance of the relays or through a relaying link. 
The relay selection is based on the maximum power gain that can 
be achieved via relaying.
The power allocated for each transmission follows 
a water-filling approach, and in case of a relaying link, 
a formula is derived to determine the power allocated to 
the source and the relay respectively.   
If constant power is allocated per channel block, the optimal
resource allocation scheme is to schedule not more than 
two users within one channel block.  
Simulation results demonstrate that our proposed 
resource allocation scheme provides considerable throughput
gains especially for users receiving low power from 
the base station.

\appendices

\section{Multiple Relay Selection with Coherent Transmission}
\label{sec:appendix}

We extend the results presented in Section~\ref{sec:solution} into the 
case when multiple relays can be involved in the RDF transmission mode
within the same resource unit (RU). In this case, all 
involved relays listen to the signal transmitted 
by the source in the first time slot, and
then decode the signal and retransmit it coherently to the destination
node in the second time slot such that the transmitted signals by the
relays add ``in-phase'' at the receiver in order to maximize the received
power at the destination $D$-node in the second time slot. 
This necessitates that in addition to the channel ``power gain'' information, 
the channels' ``phase'' information should be known as well at the transmitters.

The multiple relay selection case can be viewed as if the relay nodes represent
antennas of one ``virtual'' relay that has multiple distributed antennas.
Thus, the relay-destination link in this case is a multiple input single output
(MISO) channel. Furthermore, we can allow the source node as well to transmit
in the second slot as one of the ``virtual'' antennas of the ``virtual'' multiple-antenna
relay. The capacity and optimal resource allocation over
MISO channel is well known in the literature (e.g.~\cite{TsVi:05}).
The maximum achievable effective power gain $\gamma^{DF\_MISO}_{r_{\Omega_i[k]}d_i}[k]$ of 
the MISO channel to destination node $i$ in channel block $k$ is the sum of the 
effective power gains of the individual links
from every antenna to the $D$-node: 
\begin{equation}
\label{eq:equ_gamma}
\gamma^{DF\_MISO}_{r_{\Omega_i[k]}d_i}[k]=\gamma_{sd_i}[k]+\sum_{j\in\Omega_i[k]}\gamma_{r_jd_i}[k]
\end{equation}
where $\Omega_i[k]$ is the set of indices of involved relays in the transmission
to $D_i$ in channel block $k$.

The maximum achievable effective power gain \eqref{eq:equ_gamma} 
can be achieved by in-phase transmission
of the distributed antennas and by adjusting the power transmitted 
over every antenna $j$ to be directly proportional to its associated effective power gain
$\gamma_{r_jd}$. Thus, the power (spectral density) allocated to the relays follows:
\begin{equation}
\label{eq:MISO_power}
P^{DF\_MISO}_{r_jd_i}[k]=\frac{\gamma_{r_jd_i}[k]}{\gamma^{DF\_MISO}_{r_{\Omega_i[k]}d_i}[k]}P^{DF\_MISO}_{r_{\Omega_i[k]}d_i}[k],
\quad : j\in\Omega_i[k]
\end{equation}
where $P^{DF\_MISO}_{r_{\Omega_i[k]}d_i}[k]$ is the sum (of all antennas) power (spectral density)
allocated to the second time slot of the RDF link, and $\gamma^{DF\_MISO}_{r_{\Omega_i[k]}d_i}[k]$ 
is obtained in \eqref{eq:equ_gamma}. Similarly, the power allocated to the source
node in the second time slot of RDF link is 
$P^{DF\_MISO}_{sd_i}[k]=\frac{\gamma_{sd_i}[k]}{\gamma^{DF\_MISO}_{r_{\Omega_i[k]}d_i}[k]}P^{DF\_MISO}_{r_{\Omega_i[k]}d_i}[k]$.
All relays of the MISO RDF link should be able to decode the transmitted signal
by the source in the first time slot. Thus, the achievable rate region of
the RDF MISO link is:
\begin{align}
& R_i^{DF\_MISO}[k]= \nonumber \\
& \min\left\{
   \begin{array}{l}
        \frac{1}{2}\log\left(1+\gamma_{sr_{\Omega_i[k]}\_min}^{DF\_MISO}[k]P_{sd_i}^{DF\_MISO}[k]\right), \\
        \frac{1}{2}\log \Big( 1+\gamma_{sd_i}[k]P_{sd_i}^{DF\_MISO}[k]+ \\
        \quad \quad \quad \quad \quad \gamma^{DF\_MISO}_{r_{\Omega_i[k]}d_i}[k]P_{r_{\Omega_i[k]}d_i}^{DF\_MISO}[k] \Big)
    \end{array} 
 \right. \label{eq:MISOrelay}
\end{align}
where $\gamma_{sr_{\Omega_i[k]}\_min}^{DF\_MISO}[k]$ is the minimum effective power gain of the 
source-relay channels (among the involved relays in the MISO relay link). 
Since all relays should decode the information in the
first time slot, the achievable rate in the first time slot is bounded by the worst channel condition
among all relays.
Our objective is to find the optimal power allocation in the first and second time slots
of the RDF links, and additionally to select the relays to be involved in the
distributed MISO relay link.
We first observe the following facts:
\begin{itemize}
\item All relays that have $\gamma_{sr} < \gamma_{sd}$ should not be involved
in the MISO RDF link in order for the RDF link to support higher achievable rates
than in the case of direct transmission (DT). This fact can be easily proven similar to
\textit{Proposition $1$} in Section~\ref{sec:RDF}. 
\item \textit{Proposition $2$} in Section~\ref{sec:RDF} is also valid in the case of multiple relays selection.
This means that, for a given set ($\Omega_i$) of selected relays for the MISO RDF link, 
the power allocated to the first time slot ($P_{sd_i}^{DF\_MISO}$) and for the second 
time slot ($P_{r_{\Omega_i}d_i}^{DF\_MISO}$) respectively follow \eqref{eq:PsPrDF}
and \eqref{eq:define_acute} with the replacement of $\gamma_{sr}$ and $\gamma_{rd}$ by
$\gamma_{sr_{\Omega_i}\_min}^{DF\_MISO}$ and $\gamma^{DF\_MISO}_{r_{\Omega_i}d_i}$ respectively.
Similarly, the achievable rate can be obtained using \eqref{eq:rateRDFlink} and \eqref{eq:alpha}.
\end{itemize}

The power gain $\alpha$ that the MISO RDF link can support can be obtained, similar to \eqref{eq:alpha}, by:
\begin{equation}
\label{eq:MISO_alpha}
\alpha_i^{DF\_MISO}=\frac{\frac{\gamma_{sr_{\Omega_i}\_min}^{DF\_MISO}}{\gamma_{sd_i}}
\frac{\gamma^{DF\_MISO}_{r_{\Omega_i}d_i}}{\gamma_{sd_i}}}
{\frac{\gamma_{sr_{\Omega_i}\_min}^{DF\_MISO}}{\gamma_{sd_i}}
+\frac{\gamma^{DF\_MISO}_{r_{\Omega_i}d_i}}{\gamma_{sd_i}}-1}
\end{equation}

We aim to obtain the set of relays that maximize \eqref{eq:MISO_alpha}. It is obvious that
$\alpha$ increases (or decreases) by increasing (or decreasing) $\gamma_{sr_{\Omega_i}\_min}^{DF\_MISO}$
or $\gamma^{DF\_MISO}_{r_{\Omega_i}d_i}$. We can increase $\gamma_{sr_{\Omega_i}\_min}^{DF\_MISO}$
by removing the relay with worst channel quality with the source. However, by removing this relay from
the set of relays, we decrease $\gamma^{DF\_MISO}_{r_{\Omega_i}d_i}$ which is a sum of all relay-destination
channels gains. Thus, we can adopt the following procedure to find the set of relays that 
produce the highest possible power gain \eqref{eq:MISO_alpha}:
\begin{enumerate}
\item Start with the set of all relays that satisfy the condition $\gamma_{sr} > \gamma_{sd}$, 
and compute the achievable power gain \eqref{eq:MISO_alpha}.
\item Modify the set of relays by removing the relay that has the least $\gamma_{sr}$, 
and re-compute the power gain.
\item Repeat step $2$ until you get a set of one relay only.
\item Compare the power gains of the examined sets of relays' selection and choose the one 
that support the highest power gain. 
\end{enumerate}

Note that the results of Section~\ref{sec:optimal} and Section~\ref{sec:suboptimal} are also valid for 
the case of MISO RDF.


\bibliographystyle{IEEEbib}
\bibliography{./BIBFILES/papers,./BIBFILES/books}

\begin{thebibliography}{10}

\bibitem{Va:1971}
E.~C.~Van~Der Meulen,
\newblock ``Three-terminal communication channels,''
\newblock {\em Advances in Applied Probability}, vol. 3, pp. 120--154, 1971.

\bibitem{COEL:1979}
T.~Cover and A.~El Gamal,
\newblock ``Capacity theorems for the relay channel,''
\newblock {\em IEEE Transactions on Information Theory}, vol. 25, no. 5, pp.
  572--584, Sept. 1979.

\bibitem{NiTsWo:2004}
J.~Nicholas Laneman, D.~Tse, and G.~Wornell,
\newblock ``Cooperative diversity in wireless networks: Efficient protocols and
  outage behavior,''
\newblock {\em IEEE Transactions on Information Theory}, vol. 50, no. 12, pp.
  3062--3080, Dec. 2004.

\bibitem{LaLiEl:2006}
L.~Lai, K.~Liu, and H.~El Gamal,
\newblock ``The three node wireless network: Achievable rates and cooperation
  strategies,''
\newblock {\em IEEE Transactions on Information Theory}, vol. 52, no. 3, pp.
  805--828, Mar. 2006.

\bibitem{KrGaGu:2005}
G.~Kramer, M.~Gastpar, and P.~Gupta,
\newblock ``Cooperative strategies and capacity theorems for relay networks,''
\newblock {\em IEEE Transactions on Information Theory}, vol. 51, no. 9, pp.
  3037--3063, Sept. 2005.

\bibitem{ElMoZa:2006}
A.~El Gamal, M.~Mohseni, and S.~Zahedi,
\newblock ``Bounds on capacity and minimum energy-per-bit for {AWGN} relay
  channels,''
\newblock {\em IEEE Transactions on Information Theory}, vol. 52, no. 4, pp.
  1545--1561, Apr. 2006.

\bibitem{HoZh:2005}
A.~{Host-Madsen} and J.~Zhang,
\newblock ``Capacity bounds and power allocation for wireless relay channels,''
\newblock {\em IEEE Transactions on Information Theory}, vol. 51, no. 6, pp.
  2020--2040, June 2005.

\bibitem{ReKuVe:2004}
A.~Reznik, S.~R. Kulkarni, and S.~Verdu,
\newblock ``Degraded {G}aussian multirelay channel: Capacity and optimal power
  allocation,''
\newblock {\em IEEE Transactions on Information Theory}, vol. 50, no. 12, pp.
  3037--3046, Dec. 2004.

\bibitem{GuEr:2007}
D.~Gunduz and E.~Erkip,
\newblock ``Opportunistic cooperation by dynamic resource allocation,''
\newblock {\em IEEE Transactions on Wireless Communications}, vol. 6, no. 4,
  pp. 1446--1454, Apr. 2007.

\bibitem{StEr:2004}
A.~Stefanov and E.~Erkip,
\newblock ``Cooperative coding for wireless networks,''
\newblock {\em IEEE Transactions on Communications}, vol. 52, no. 9, pp.
  1470--1476, Sept. 2004.

\bibitem{YaHuXuMa:2009}
Y.~Yang, H.~Hu, J.~Xu, and G.~Mao,
\newblock ``Relay technologies for {WiMAX} and {LTE-Advanced} mobile systems,''
\newblock {\em IEEE Communications Magazine}, vol. 47, no. 10, pp. 100--105,
  Oct. 2009.

\bibitem{SeToBa:2009}
S.~Sesia, I.~Toufik, and M.~Baker,
\newblock {\em {LTE} -– The {UMTS} Long Term Evolution: From Theory to
  Practice},
\newblock John Wiley \& Sons, Ltd., 2009.

\bibitem{DaPaSkBe:2008}
E.~Dahlman, S.~Parkvall, J.~Sk\"old, and P.~Beming,
\newblock {\em 3{G} Evolution: {HSPA} and {LTE} for Mobile Broadband},
\newblock Academic Press, Elsevier, second edition, 2008.

\bibitem{EkFuKaMePaToWa:2006}
H.~Ekstrom, A.~Furuskar, J.~Karlsson, M.~Meyer, S.~Parkvall, J.~Torsner, and
  M.~Wahlqvist,
\newblock ``Technical solutions for the 3{G} long-term evolution,''
\newblock {\em IEEE Communications Magazine}, vol. 44, no. 3, pp. 38--45, Mar.
  2006.

\bibitem{Ek:2009}
H.~Ekstrom,
\newblock ``{QoS} control in the {3GPP} evolved packet system,''
\newblock {\em IEEE Communications Magazine}, vol. 47, no. 2, pp. 76--83, Feb.
  2009.

\bibitem{TsVi:05}
D.~Tse and P.~Viswanath,
\newblock {\em Fundamentals of Wireless Communication},
\newblock Cambridge University Press, May 2005.

\bibitem{Gal:68}
R.~G. Gallager,
\newblock {\em Information Theory and Reliable Communication},
\newblock John Wiley \& Sons, Inc., 1968.

\bibitem{TsHa:98}
D.~Tse and S.~Hanly,
\newblock ``Multiaccess fading channels -- {P}art 1: Polymatroid structure,
  optimal resource allocation and throughput capacities,''
\newblock {\em IEEE Transactions on Information Theory}, vol. 44, no. 7, pp.
  2796--2815, Nov. 1998.

\bibitem{LiGo1:2001}
L.~Li and A.~Goldsmith,
\newblock ``Capacity and optimal resource allocation for fading broadcast
  channels -- {P}art 1: Ergodic capacity,''
\newblock {\em IEEE Transactions on Information Theory}, vol. 47, no. 3, pp.
  1083--1102, Mar. 2001.

\bibitem{Ts:98}
D.~Tse,
\newblock ``Optimal power allocation over parallel {G}aussian broadcast
  channels,''
\newblock {\em unpublished, available at
  {www.eecs.berkeley.edu/$\sim$dtse/broadcast2.pdf}}.

\bibitem{LiVe:2005}
Y.~Liang and V.~Veeravalli,
\newblock ``Gaussian orthogonal relay channels: Optimal resource allocation and
  capacity,''
\newblock {\em IEEE Transactions on Information Theory}, vol. 51, no. 9, pp.
  3284--3289, Sept. 2005.

\bibitem{LiVeVi:2007}
Y.~Liang, V.~Veeravalli, and H.~Vincent Poor,
\newblock ``Resource allocation for wireless fading relay channels: Max-min
  solution,''
\newblock {\em IEEE Transactions on Information Theory}, vol. 53, no. 10, pp.
  3432--3453, Oct. 2007.

\bibitem{ReKuVe:2005}
A.~Reznik, S.~R. Kulkarni, and S.~Verdu,
\newblock ``Broadcast-relay channel: Capacity region bounds,''
\newblock in {\em Proceedings IEEE International Symposium on Information
  Theory (ISIT)}, Sept. 2005, pp. 820--824.

\bibitem{LiVe:2007}
Y.~Liang and V.~Veeravalli,
\newblock ``Cooperative relay broadcast channels,''
\newblock {\em IEEE Transactions on Information Theory}, vol. 53, no. 3, pp.
  900--928, Mar. 2007.

\bibitem{MeDa:2008}
W.~Mesbah and T.~Davidson,
\newblock ``Power and resource allocation for orthogonal multiple access relay
  systems,''
\newblock in {\em Proceedings IEEE International Symposium on Information
  Theory (ISIT)}, July 2008, pp. 2272--2276.

\bibitem{MeDa2:2008}
W.~Mesbah and T.~Davidson,
\newblock ``Power and resource allocation for orthogonal multiple access relay
  systems,''
\newblock {\em EURASIP Journal on Advances in Signal Processing, Special Issue
  on Cooperative Wireless Networks, {Article ID 476125}}, vol. 2008, pp. 1--15,
  2008.

\bibitem{ViMu:2005}
H.~Viswanathan and S.~Mukherjee,
\newblock ``Performance of cellular networks with relays and centralized
  scheduling,''
\newblock {\em IEEE Transactions on Wireless Communications}, vol. 4, no. 5,
  pp. 2318--2328, Sept. 2005.

\bibitem{KaPoHa:2009}
M.~Kaneko, P.~Popovski, and K.~Hayashi,
\newblock ``Throughput-guaranteed resource-allocation algorithms for
  relay-aided cellular {OFDMA} system,''
\newblock {\em IEEE Transactions on Vehicular Technology}, vol. 58, no. 4, pp.
  1951--1964, May 2009.

\bibitem{SaAdRaYaFaKi:2010}
M.~Salem, A.~Adinoyi, M.~Rahman, H.~Yanikomeroglu, D.~Falconer, and Y.~Kim,
\newblock ``Fairness-aware radio resource management in downlink {OFDMA}
  cellular relay networks,''
\newblock {\em IEEE Transactions on Wireless Communications}, vol. 9, no. 5,
  pp. 1628--1639, May 2010.

\bibitem{ShGo3:2008}
M.~Shaqfeh and N.~Goertz,
\newblock ``Performance analysis of scheduling policies for delay-tolerant
  applications in centralized wireless networks,''
\newblock in {\em Proceedings IEEE International Symposium on Performance
  Evaluation of Computer and Telecommunication Systems (SPECTS 2008)}, June
  2008, pp. 309--316.

\bibitem{ShGoMc:2008}
M.~Shaqfeh, N.~Goertz, and S.~McLaughlin,
\newblock ``Organizing multiuser operation in centralized wireless networks,''
\newblock in {\em Wireless World Research Forum Meeting 20}, Apr. 2008, pp.
  1--6.

\bibitem{ViTsLa:02}
P.~Viswanath, D.~Tse, and R.~Laroia,
\newblock ``Opportunistic beamforming using dumb antennas,''
\newblock {\em IEEE Transactions on Information Theory}, vol. 48, no. 6, pp.
  1277--1294, June 2002.

\bibitem{ShGo4:2008}
M.~Shaqfeh and N.~Goertz,
\newblock ``Channel-aware scheduling with resource-sharing constraints in
  wireless networks,''
\newblock in {\em Proceedings IEEE International Conference on Communications
  (ICC)}, May 2008, pp. 4149--4153.

\bibitem{SoLi3:2005}
G.~Song and Y.~Li,
\newblock ``Utility-based resource allocation and scheduling in {OFDM}-based
  wireless broadband networks,''
\newblock {\em IEEE Communications Magazine}, vol. 43, no. 12, pp. 127--134,
  Dec. 2005.

\bibitem{WaGiMa:2007}
X.~Wang, G.~Giannakis, and A.~Marques,
\newblock ``A unified approach to {QoS}-guaranteed scheduling for
  channel-adaptive wireless networks,''
\newblock {\em Proceedings of the IEEE}, vol. 95, no. 12, pp. 2410--2431, Dec.
  2007.

\bibitem{WaGi2:2007}
X.~Wang and G.~Giannakis,
\newblock ``Stochastic primal-dual scheduling subject to rate constraints,''
\newblock in {\em Proceedings IEEE Wireless Communications and Networking
  Conference (WCNC)}, Mar. 2007, pp. 1527--1531.

\bibitem{Pro:95}
J.~G. Proakis,
\newblock {\em Digital Communications},
\newblock McGraw-Hill International Editions, third edition, 1995.

\bibitem{Go:2005}
A.~Goldsmith,
\newblock {\em Wireless Communications},
\newblock Cambridge University Press, 2005.

\end{thebibliography}

\end{document}